\documentclass[aps,pra,twocolumn,showpacs]{revtex4-1}

\usepackage{graphicx,epstopdf}
\usepackage{amsmath}
\usepackage{amssymb}
\usepackage{braket}
\usepackage{color}
\usepackage{float}
\usepackage{epstopdf}
\usepackage{amsmath}
\usepackage{amsthm}
\usepackage{braket}
\usepackage{bm}
\usepackage{graphics,epstopdf}
\usepackage{color}
\usepackage[usenames,dvipsnames,svgnames]{pstricks}
\usepackage{epsfig}
\usepackage{pst-grad} 
\usepackage{pst-plot}

\newtheorem{thm}{Theorem}

\usepackage[colorlinks=true, citecolor=blue, urlcolor=blue ]{hyperref}
\input{epsf}
\begin{document}

%
%

\title{Stronger Error Disturbance Relations for Incompatible Quantum Measurements}
\author{Chiranjib Mukhopadhyay, Namrata Shukla, and Arun Kumar Pati} 
\affiliation{Harish-Chandra Research Institute, Chhatnag Road, Jhunsi, Allahabad 211
019, India} 

\date{\today}

\begin{abstract}
We formulate a new error-disturbance relation, which is free from explicit 
dependence upon variances in observables. This error-disturbance relation shows improvement over the one provided by the Branciard inequality 
and the Ozawa inequality for some initial states and for particular class of joint measurements under consideration. We also prove a modified form of 
Ozawa's error-disturbance relation. The later relation provides a tighter bound compared to the Ozawa and the Branciard inequalities for a small number of states.
\end{abstract}

\maketitle

\begin{section}{Introduction}
Uncertainty principle first enunciated by Heisenberg \cite{Heisenbergoriginal} is one of the basic tenets 
of quantum mechanics and still a subject of active investigation. 
The original uncertainty principle encapsulates the impossibility of simultaneous measurement of two incompatible physical
observables with arbitrary precision, as the measurement of one disturbs the other. 
Heisenberg also gave, what he thought was a mathematical formulation of this principle for position 
and momentum operators. This was later put on firm footing for general physical 
observables by Kennard \cite{Kennard}. However, the uncertainty relation was rigorously proved by 
Robertson \cite{Robertson} 
and tightened by Schr{\"o}dinger \cite{Schroedinger}. 
These relations were collectively called uncertainty relations. The
Robertson version of the uncertainty relation is given by
\begin{equation}
\Delta A  \Delta B \geq \frac{1}{2}|\langle\psi|[A,B]|\psi\rangle|, \label{Robertson}
\end{equation}
where $A$ and $B$ are two incompatible observables on the Hilbert space of the system and the 
variance of an observable $A$ in a quantum state $|\psi\rangle$ is given by $\Delta A^2 = \langle\psi|A^{2}|\psi\rangle - {\langle\psi|A|\psi\rangle}^2$. Uncertainty relations are 
of great importance in physics including foundations of 
quantum mechanics, quantum information and can have various technological 
applications \cite{Busch,Hall1,Hofman,Guehne,Fuchs}. 
It may be noted that uncertainty relations given before may happen 
to be trivial even if the observables are incompatible on the state of system. This problem was recently 
cured with the introduction of stronger uncertainty relations by Maccone and Pati \cite{Macconepati} and these 
capture the concept of incompatible observables.\\

Arthurs and Kelly \cite{Arthur-kelly} derived an expression akin to the 
Robertson uncertainty relation for error in measurement of observable $A$ and corresponding 
disturbance on observable $B$. However, this was shown to be violated by Arthurs and Goodman for unbiased measurements \cite{Arthur-Goodman}. 
Later, Ozawa proved a relation \cite{Ozawa, Ozawa2}, which connects error in measuring one observable
and corresponding disturbance in another observable to the quantum fluctuations (variances) of these two incompatible 
observables. The error and disturbance mentioned here contain the information about 
interactions of the system with the measuring apparatus.
Ozawa showed that for general measurement strategies instead of unbiased measurements, the bound given in Ref.\cite{Arthur-Goodman}
can be violated. This was verified experimentally by Rozema {\it et al.} \cite{Rozema} and Erhart {\it et al.} \cite{Erhart}. \\


There is a recent debate on alleged violations of the Heisenberg
error-disturbance relation \cite{Buschlahtiwerner, Ozawaresolve, Buscemi, Busch01, Busch02}. This debate originates  from the fact that two
approaches start from different definitions of “error” and “disturbance”.
For example, Ozawa's approach \cite{Ozawaresolve} is based on the expectations of squared
differences of noise operators in a measurement
process. These quantities depend on the input state of the measurement
apparatus. However, the definitions in Busch-Lahti-Werner \cite{Buschlahtiwerner} approach are
characteristic of the
measurement scheme, and hence independent of the input state. They have
shown that the standard textbook form of uncertainty relation is still
respected. One should note that in these two approaches, error and
disturbance quantifiers have different meanings.
In addition to several investigations on measurement related uncertainty \cite{Busch1, Busch2, Busch3, Busch4,Macccone}, there have been subsequent developments 
on improving the tightness of the bound provided by Ozawa's 
error-disturbance relation \cite{Hall, Weston, Branciard}. 
The Branciard bound in this series of error-disturbance relations is 
known to be tight compared to the Ozawa relation \cite{Qiao}. However we will not dwell on the debate here.\\ 

In this letter, we intend to prove a new stronger error-disturbance relation for incompatible quantum measurement, 
that does not depend on the variances of observables. This
new error-disturbance relation provides a stronger bound than the Branciard bound for some initial states. 
We derive another error disturbance relation 
which is the modified form of Ozawa's error-disturbance relation and is obtained by using the product of variance form of newly introduced uncertainty relations \cite{Macconepati}.
We also prove yet another stronger error-disturbance inequality for general incompatible observables.\\

\end{section}

\begin{section}{Error-disturbance relations}
Let us consider the 
Heisenberg picture and treat quantum states as time independent, {\it i.e.}, the effect of interaction being manifested 
through the evolution of Hermitian operators which are physical observables of the system under consideration. 
We assume that the system and the apparatus (probe) are initially non-entangled and represented by states
$|\psi\rangle_{s}$ and $|\phi\rangle_{p}$, respectively. The physical observables that we want to measure are
$A$ and $B$ such that in the joint Hilbert space $ \mathcal{H}_s\otimes\mathcal{H}_p$, 
we have $A_{in} = A \otimes \mathbb{I}$ and $B_{in} = B \otimes \mathbb{I}$. We now fix an 
operator $M$ given by $M_{in} = \mathbb{I} \otimes M$ on the Hilbert space of the probe. We will use this operator (after measurement) 
to read off and estimate the value of $A$. An entangling global unitary $U$ can be used to couple the system 
to the probe and this interaction transforms these aforementioned 
operators into new operators $A_{out} = U^{\dagger}\left(A \otimes \mathbb{I}\right) U$,
$B_{out} = U^{\dagger}\left(B \otimes \mathbb{I}\right) U$, and
$M_{out} = U^{\dagger}\left(\mathbb{I} \otimes M\right) U$.
It is to be noted that $M$ and $B$ initially act on different Hilbert spaces 
and they remain to be commuting after the unitary evolution. Since, $M_{out}$ and $B_{out}$ are commuting, 
we expect them to be simultaneously measurable. Therefore, the problem of impossibility of joint measurements in this measurement process is negated, the price 
being paid is the statistical error of estimation introduced while trying to estimate $A$ from another observable $M$. 
Now, we try to estimate this error, the natural choice being the root mean squared value of the difference between 
the estimator ($M_{out}$) and the original ($A_{in}$) observables. This is defined as noise $\epsilon_{A}$ 
in the measurement of observable $A$ and is given by \cite{Ozawa}
\begin{equation}\label{noise}
\epsilon_{A} = \sqrt{\langle\Psi|\left(M_{out}-A_{in}\right)^{2}|\Psi\rangle},
\end{equation}
where $|\Psi\rangle=|\psi\rangle_s\otimes|\phi\rangle_p$.
If the observable $B$ is measured immediately after $A$, there will be some disturbance in the measurement 
due to the prior interaction happened in the system during the measurement of the observable $A$. Thus, similar to the noise, the disturbance for 
$B$ is defined as the root mean squared value of the difference between the original observable $B_{in}$ and the 
transformed observable $B_{out}$, {\it i.e.}, 
\begin{equation}\label{disturbance}
\eta_{B} = \sqrt{\langle\Psi|\left(B_{out}-B_{in}\right)^{2}|\Psi\rangle}.
\end{equation} With these definitions at hand one would like to find a relation between the error and the disturbance. The first real improvement on this front 
for an unbiased estimator was given by Ozawa \cite{Ozawa} which reads as
\begin{equation}
\epsilon_{A}\eta_{B} + \epsilon_{A} \Delta B + \Delta A \eta_{B} \geq |\mathcal{C}_{AB}|. \label{ozawa}
\end{equation}
where $|\mathcal{C}_{AB}|=\frac{1}{2}|\langle\psi|[A,B]|\psi\rangle|$. Note that, the Ozawa relation 
depends on $\Delta A$ as well as $\Delta B$.\\

After Ozawa's work, there has been many more error-disturbance relations given by different authors \cite{Hall,Weston,Branciard}.
The Branciard error-disturbance relation acclaimed as the best of them all, is expressed as

\begin{equation}
\epsilon_{A}^{2} \Delta B^{2} + \eta_{B}^{2} \Delta A^{2} + 2 \epsilon_{A} \eta_{B} \sqrt{\Delta A^{2} \Delta B^{2} - \mathcal{C}_{AB}^{2}} \geq \mathcal{C}_{AB}^{2}.\label{branciard1}
\end{equation}
However, all the existing error-disturbance relations do involve variances of $A, B, B_{out}$ and $M_{out}$.
We look for an error-disturbance relation free from quantum fluctuations in the observables. 
This is the new feature of our main error-disturbance relation.
\end{section}



\begin{section}{New error-disturbance Relations}

In this section, we derive two different error-disturbance relations and illustrate their efficiencies with some examples. 

\begin{thm}\label{newmdr} 
For the noise operator $N_{A}=M_{out} - A_{in}$ and corresponding disturbance operator $ D_{B} = B_{out} - B_{in}$, 
if the system and the probe are in joint state
$|\Psi\rangle = |\psi\rangle_{s} \otimes |\phi\rangle_{p}$, the following inequality holds:
\begin{eqnarray}\label{firstmdr}
\epsilon_{A}^{2} + \eta_{B}^{2} &\geq &\pm i\langle \psi|[A,B]|\psi \rangle \mp i\langle \Psi|[M_{out},B_{in}]|\Psi\rangle \nonumber\\ 
& & \mp i\langle\Psi|[A_{in},B_{out}]|\Psi\rangle +|\langle\Psi | N_{A} \pm i D_{B} |\Psi^{\perp} \rangle |^{2}, \nonumber\\
\label{firstmdreqn}
\end{eqnarray}
where the sign is chosen such that $\pm i\langle\psi|[A, B]|\psi\rangle $ is positive (and similarly for other commutators) 
and $|\Psi^{\perp}\rangle$ is orthogonal to $|\Psi\rangle$.
\end{thm}
\begin{proof}

For the above mentioned observables $N_A$ and $D_B$, we define, $ C = N_A -\langle N_A\rangle $ and $ D = D_B -\langle D_B\rangle $, where 
$\langle N_A\rangle=\langle\Psi|N_A|\Psi\rangle,~\langle D_B\rangle=\langle\Psi|D_B|\Psi\rangle$.
The standard deviations of $N_A$ and $D_B$ can therefore be written as $\Delta N_A = \|C|\Psi\rangle\|$ and
$\Delta D_B = \|D|\Psi\rangle\|$. Consider the quantity 
\begin{equation} \label{interm}
\| (C \mp iD)|\Psi\rangle\|^{2}  = \Delta N_A^{2} + \Delta D_B^{2} \mp i \langle \Psi|[N_A,D_B]|\Psi\rangle.
\end{equation}
Using the Cauchy-Schwarz inequality, the LHS of this equation is bounded from below as

\begin{eqnarray} \label{interm2}
&&\left| \langle \Psi | N_A \pm iD_B |\Psi^{\perp}\rangle \right|^{2}\nonumber\\
&=& \left| \langle\Psi | (N_A \pm iD_B) -\langle N_A \pm iD_B \rangle |\Psi^{\perp}\rangle \right|^{2} \nonumber \\
&=&\left| \langle \Psi | C \pm iD |\Psi^{\perp}\rangle \right|^{2} \nonumber \\
&\leq & \| (C \mp iD)|\Psi\rangle\|^{2}.
\end{eqnarray}
Combining Eq.~(\ref{interm}) and Eq.~(\ref{interm2}) leads to

\begin{eqnarray}\label{newmdrprep}
&&\Delta N_{A}^{2} + \Delta D_{B}^{2}\nonumber\\
&\geq & \pm i \langle \Psi |[N_{A},D_{B}]|\Psi \rangle + \vert\langle \Psi |N_{A} \pm iD_{B}|\Psi^{\perp} \rangle \vert^{2} \nonumber \\
                                & = &  \pm i \langle\Psi|[M_{out},B_{out}]|\Psi \rangle \pm i \langle \psi|[A,B]|\psi \rangle \nonumber\\
                                & & \mp i \langle \Psi|[M_{out},B_{in}]|\Psi\rangle \mp i \langle \Psi|[A_{in},B_{out}]|\Psi \rangle \nonumber \\
                                & &  + \vert\langle \Psi |N_{A} \pm iD_{B}|\Psi^{\perp} \rangle \vert^{2} \nonumber \\
                                & = &  \pm i \langle \psi|[A,B]|\psi \rangle \mp i \langle \Psi|[M_{out},B_{in}]|\Psi\rangle \nonumber \\
                                & & \mp i \langle\Psi|[A_{in},B_{out}]|\Psi\rangle + \vert\langle \Psi |N_{A} \pm iD_{B}|\Psi^{\perp} \rangle \vert^{2}\nonumber\\
\end{eqnarray}
The new error-disturbance relation for the sum of squares of noise and disturbance follows from the definitions
in Eq.~(\ref{noise}) and 
Eq.~(\ref{disturbance}), and using the fact that $\epsilon_{A}^{2}\geq \Delta N_A^2,~\eta_{B}^{2}\geq \Delta D_B^2$ and $[M_{out}, B_{out}]=0$, {\it i.e.}, we have
\begin{eqnarray}
\epsilon_{A}^{2} + \eta_{B}^{2}& &\geq \pm i \langle \psi|[A,B]|\psi\rangle \mp i \langle\Psi|[M_{out},B_{in}]|\Psi\rangle \nonumber \\
                                && \mp i \langle\Psi|[A_{in},B_{out}]|\Psi\rangle + \vert\langle \Psi |N_{A} \pm iD_{B}|\Psi^{\perp} \rangle \vert^{2}.\nonumber
\end{eqnarray}
Hence the proof.
\end{proof}

A noticeable point in this error-disturbance relation is, it involves no mention of variances of observables in input states and also
provides us a better bound than the Branciard bound for some choices of the initial state of the system and the measurement strategy.
To illustrate the new error-disturbance relation (\ref{firstmdr}) for a qubit system, we assume the system and the probe are initially in the states
\begin{eqnarray}
|\psi\rangle_s&=&\alpha|0\rangle+\beta|1\rangle=u|0\rangle, \nonumber\\
|\phi\rangle_p&=&|1\rangle. 
\label{initialsp}
\end{eqnarray}
where $\alpha$ is real and unitary $u$ is written as
\begin{equation} \label{unitary}
u = \left(\begin{array}{cc}
\alpha & -\beta^* \\
\beta  & \alpha \\
\end{array}\right). \nonumber
\end{equation}
We fix the input observables and estimator as
$A_{in}=\sigma_{x}^{\prime}\otimes \mathbb{I}$, $B_{in}= \sigma_{y}^{\prime}\otimes \mathbb{I}$ and $M_{in}=\mathbb{I} \otimes \sigma_{x} \label{inputM0}$, 
where $\sigma_{x}^{\prime}=u\sigma_{x}u^{\dagger},~\sigma_{y}^{\prime}=u \sigma_{y}u^{\dagger}$ and $\sigma_{x}, \sigma_{y}$ and $\sigma_{z}$ are Pauli spin matrices. 
We now couple the system and the probe through a CNOT interaction given by
$U = P_{0}\otimes \mathbb{I} + P_{1} \otimes \sigma_{x}$ with $P_0=|0\rangle\langle0|$ and $P_1=|1\rangle\langle1|$.
For the present strategy, the operators $A_{in}$, $B_{in}$ and $M_{in}$ during the interaction transform 
into $A_{out}$, $B_{out}$, and $M_{out}$. Thus, the value of the commutator term is $|\mathcal C_{AB}|=1$. Since, we have $ \Delta A=\Delta B=1$ for this particular choice of input observables,
we can express the Branciard bound using Eq.~(\ref{branciard1}) and it is given by the relation
\begin{equation}
\epsilon_A^2+\eta_B^2\geq1.
\end{equation}
Now, calculating each term in our error-disturbance relation
given by Eq.~(\ref{firstmdr}) with the above choices of input observables, 
estimator, CNOT interaction and the positive sign of $\pm i\langle\psi|[A, B]|\psi\rangle $, reduces the inequality to 
\begin{equation}
\epsilon_A^2+\eta_B^2\geq\ 2+2\alpha^2(\beta^{*}-\beta)^2+4|\alpha^2\beta^{*}(\beta^{*}-\beta)|^2.
\label{new comp}
\end{equation}
On decomposition, $ \alpha=\cos\theta,~\beta=\sin\theta~e^{i\phi}$ we have the above inequality
\begin{equation}
\epsilon_A^2+\eta_B^2\geq\ 2-8\cos^2\theta\sin^2\theta\sin^2\phi+16\cos^4\theta\sin^4\theta\sin^2\phi.
\label{new comp 1}
\end{equation}
Further, on choosing $\phi=\pi/2$ in this case of example, we have the inequality as given by
\begin{equation}
\epsilon_A^2+\eta_B^2\geq\ 1+\cos^42\theta
\label{new comp 2}
\end{equation}
We can easily see that our bound is better than the Branciard bound for this specific setting in the admissible range of $\theta$
and $\phi$. \\


\begin{figure}
\centering
\includegraphics[scale=0.90]{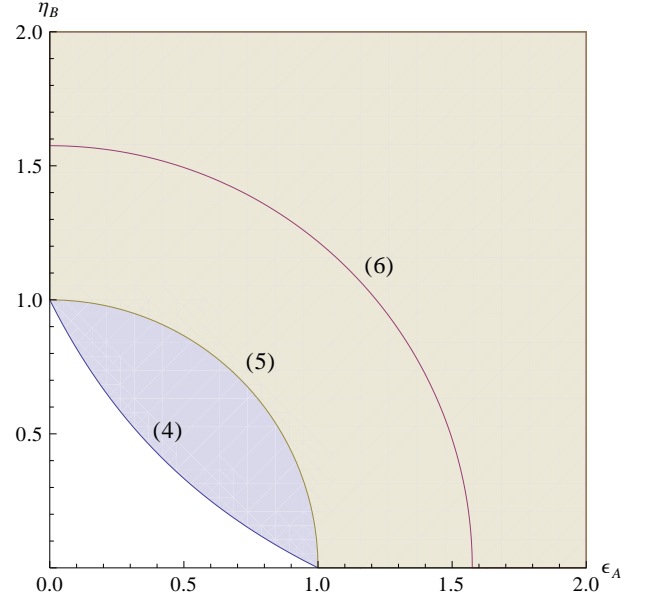}
\caption{(Color Online) Error-disturbance relations for the fixed values of observables and state such that $\mathcal C_{AB}=1$. 
The purple line denotes the best bound available from Eq.~(\ref{firstmdreqn}). The olive line is the Branciard 
bound given by Eq.~(\ref{branciard1}) and the blue line shows the bound available from generalized error-disturbance relation
Eq.~(\ref{ozawa}) by Ozawa.}
\label{Fig_6}
\end{figure}

In order to better understand the improvement of bounds provided by Ozawa \cite{Ozawa} and Branciard \cite{Branciard}, we plot all the three
error-disturbance inequalities given by Eq.~(\ref{ozawa}), Eq.~(\ref{branciard1}) and Eq.~(\ref{firstmdreqn})in the $\epsilon_A-\eta_B$ plane
for a qubit state as an example. We fix the input observables and the qubit state in such a way that $ |\mathcal C_{AB}|=1$. This can be achieved by choosing
$A_{in}=\sigma_{x}^{\prime}\otimes \mathbb{I}$, $B_{in}= \sigma_{y}^{\prime}\otimes \mathbb{I}$. We choose 
the interaction unitary to be a qubit CNOT gate $U = P_{0}\otimes \mathbb{I} + P_{1} \otimes \sigma_{x}$ with $P_0=|0\rangle\langle0|$ and $P_1=|1\rangle\langle1|$
in order to join the probe and the system. Since, in the Schr{\"o}dinger picture, states are time dependent and $|\Psi^{\perp}\rangle$ can be generated 
by projecting any state $|r\rangle $ to the orthogonal subspace of $|\Psi\rangle$, {\it i.e.},
$|\Psi^{\perp}\rangle \propto \left(\mathbb{I} - |\Psi \rangle \langle \Psi| \right) |r\rangle$, 
where $|r\rangle$ is state of system and probe.
However, we will be using the Heisenberg 
picture here and it can be easily shown that in the Heisenberg picture, we have
\begin{equation}
|\Psi^{\perp}\rangle \propto U^{\dagger}\left(\mathbb{I} - U|\Psi \rangle \langle
\Psi|U^{\dagger} \right)U|r\rangle.
\end{equation}
Maximizing the value of the final term in Eq.~(\ref{firstmdreqn}) (in order to get the best lower bound),
requires maximization over all random states $|r\rangle$. We do this maximization by randomly choosing states 
in numerics. It is evident from the Fig. \ref{Fig_6} that the new error disturbance relation introduced 
in this paper gives improvement over the existing bounds.

We should add that the Branciard relation is universally valid, {\it i.e.}, it is independent of the way the joint measurement is approximated.
However, our relations are not independent of the joint measurement approximation, as, {\it e.g.}, $ M_{out}, B_{out}$ appear explicitly. As
a consequence, relation given by Eq.~(\ref{firstmdreqn}) is only valid for 
the particular class of joint measurements.\\

For more insight of the comparison of the error-disturbance relations by Ozawa, Branciard and the authors given by Eq.~(\ref{ozawa}), Eq.~(\ref{branciard1})
and Eq.~(\ref{firstmdreqn}), respectively, we denote the left hand sides of 
these three equations which would be bounded below by the value of $|\mathcal C_{AB}|$, respectively as given by 

\begin{equation}
L_{Ozawa}=\epsilon_{A}\eta_{B} + \epsilon_{A} \Delta B + \Delta A \eta_{B},\label{L Ozawa}
\end{equation}
\begin{eqnarray}
&&L_{Branciard}\nonumber\\
&=&\sqrt{\epsilon_{A}^{2} \Delta B^{2} + \eta_{B}^{2} \Delta A^{2}
+2 \epsilon_{A} \eta_{B} \sqrt{\Delta A^{2} \Delta B^{2}-\mathcal{C}_{AB}^{2}}},\nonumber\\
\label{L Branciard} 
\end{eqnarray}
and
\begin{eqnarray}
L_{New}^{(1)}&=&\frac{1}{2}\Big[\epsilon_{A}^{2} + \eta_{B}^{2} \pm i\langle \Psi|[M_{out},B_{in}]|\Psi\rangle \nonumber\\
&&\pm i\langle\Psi|[A_{in},B_{out}]|\Psi\rangle-|\langle\Psi | N_{A} \pm i D_{B} |\Psi^{\perp} \rangle |^{2}\Big],\nonumber\\
\end{eqnarray}\label{L New 1} while the sign in Eq.~(\ref{firstmdr}) is chosen that $\pm i\langle\psi|[A,B]|\psi\rangle$ is positive. 

For illustration, we give an example, where the system and the probes are two qubits. 
Let the system be initially in the state $|\psi\rangle_{s} =\cos \theta|0\rangle + \sin\theta|1\rangle$ 
and the probe be initially in the state $|\phi\rangle_{p}=|1\rangle$. We fix our input observables 
and the estimator as $ A_{in}=\sigma_{x} \otimes \mathbb{I}$, $B_{in} = \sigma_{y} \otimes \mathbb{I}$ 
and $M_{in} = \mathbb{I} \otimes \sigma_{x}$.
The system is again coupled with the probe through a CNOT interaction and we use the same method to 
generate the $|\Psi^{\perp}\rangle$  as in Heisenberg picture, we have $
|\Psi^{\perp}\rangle \propto U^{\dagger}\left(\mathbb{I} - U|\Psi \rangle \langle
\Psi|U^{\dagger} \right)U|r\rangle$. In order to get the best lower bound the maximization is performed 
by randomly choosing states $|r\rangle$ in numerics. One can also see that, $B_{out}$ and $B_{in}$ have the same spectrum, 
but not $M_{out}$ and $A_{in}$. Note that, the best bound given by the Branciard inequality in this case is obtained on replacing 
$\eta_{B}$ by $\eta_{B} \sqrt{1-\frac{\eta_{B}^{2}}{4}}$ in Eq.~(\ref{branciard1})\cite{Branciard}. In Fig. \ref{Fig_2}, we use this bound 
while comparing with the new relation. It is evident in Fig. \ref{Fig_2} that for the above mentioned choices of initial states of system, probe and 
their interaction, the bound presented here goes beyond the tightest possible Branciard bound 
for approximately 25$\%$ of the states. Unfortunately, these points do not fall on a clearly visible line, 
signifying that finding analytical expression for this bound in simpler terms is challenging.\\



\begin{figure}
\centering
\includegraphics[scale=0.32,angle = 270]{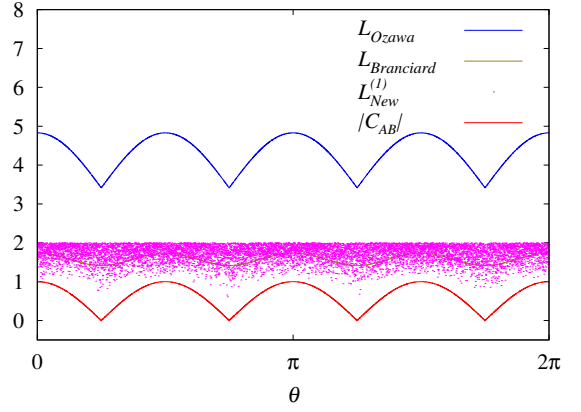}
\caption{(Color Online) The red line is the value of the commutator for any arbitrary initial system state (qubit), 
$\theta \in [0,2\pi]$ (for 10000 different states). The olive line is the tightest possible Branciard 
bound and the purple dots are points corresponding to the new error-disturbance relation inequality Eq.~(\ref{firstmdreqn}). 
The blue line corresponds to Ozawa's error-disturbance relation. 
It is seen that for roughly 25$\%$ of states, the new bound is tighter than the Branciard bound.} 
\label{Fig_2}
\end{figure}

\begin{thm}\label{modifiedozawapdt}
For Noise operator $N_{A}$ and corresponding Disturbance operator $D_{B}$ defined as, $N_{A} = M_{out} - A_{in}$ 
and $ D_{B} = B_{out} - B_{in}$, if the system and the probe are in joint state 
$|\Psi\rangle = |\psi\rangle_{s} \otimes |\phi\rangle_{p}$, the following inequality can be proved:
\begin{eqnarray}
    & \epsilon_{A}\eta_{B} + \eta_{B}\Delta A + \epsilon_{A}\Delta B & \nonumber \\
                 & -\frac{1}{2} \frac{|\langle\Psi | N_{A}\Delta D_{B}  \pm i D_{B} \Delta N_{A}|\Psi^{\perp} \rangle |^{2} }{\epsilon_{A}\eta_{B}} & \nonumber \\
                 & - \frac{1}{2} \frac{|\langle\Psi | A \Delta D_{B}  \pm i D_{B} \Delta A|\Psi^{\perp} \rangle |^{2} }{\Delta A\eta_{B}} & \nonumber \\
                 &  - \frac{1}{2} \frac{|\langle\Psi | N_{A}\Delta B  \pm i B\Delta N_{A}|\Psi^{\perp} \rangle |^{2} }{\epsilon_{A}\Delta B} &\geq |\mathcal{C}_{AB}|. 
\label{modiozawapdt}
\end{eqnarray}

\end{thm}

\begin{proof}
For two arbitrary observables $A$ and $B$, the following uncertainty relation \cite{Macconepati} is
satisfied
\begin{equation}\label{lemma2}
\Delta A \Delta B \geq \frac{\pm\frac{i}{2}\langle \Psi|[A,B]|\Psi
\rangle}{1-\frac{1}{2}|\langle \Psi |\frac{A}{\Delta A} \pm i \frac{B}{\Delta B}
|\Psi^{\perp}\rangle|^{2}},
\end{equation}
where variables and averages are defined in the state $|\Psi\rangle$. For arbitrary states, $|\Psi^{\perp}\rangle$ is orthogonal to the state of the
system $\vert\Psi \rangle$ , and the sign is chosen such that $\pm i\langle\Psi|[A,B]|\Psi\rangle $ is positive.
To prove the inequality (\ref{modiozawapdt}) we note that 
\begin{equation}
-[A\otimes\mathbb{I},B\otimes\mathbb{I}]=[N_{A},D_{B}] + [N_{A},B] + [A,D_{B}]. \nonumber
\end{equation}
This implies, we have
\begin{eqnarray}
&&|\langle\Psi|[A,B]|\Psi\rangle| \nonumber\\
&&=\vert \langle\Psi|[N_{A},D_{B}]|\Psi\rangle+ \langle\Psi|[N_{A},B]|\Psi\rangle + \langle\Psi|[A,D_{B}]|\Psi\rangle \vert \nonumber \\
&&\leq \vert\langle\Psi|[N_{A},D_{B}]|\Psi\rangle \vert + \vert\langle\Psi|[N_{A},B]|\Psi\rangle\vert + \vert\langle\Psi|[A,D_{B}]|\Psi\rangle\vert. \nonumber
\label{commutatorozawa}
\end{eqnarray}

\begin{figure}
\includegraphics[scale=0.32,angle = 270,origin = c]{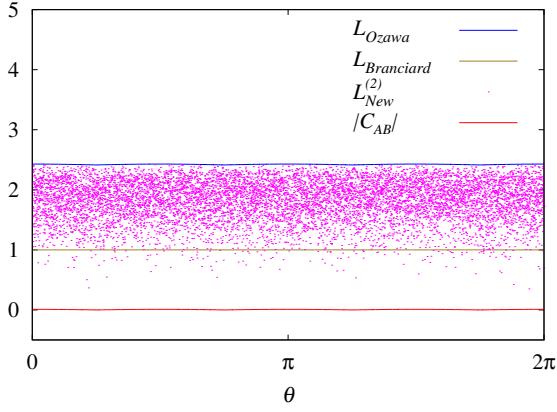}
\caption{(Color Online) The red line is the value of the commutator for any arbitrary 
initial system state (qubit), $\theta \in [0,2\pi]$ (for 10000 different states). 
The olive line is the tightest possible Branciard 
bound and the purple dots are points corresponding to the new error-disturbance relation Eq.~(\ref{modiozawapdt}). 
The blue line corresponds to Ozawa's error-disturbance relation. 
It is seen that  the new bound is tighter than the Branciard bound for roughly $1.5\%$ states.}
\label{Fig_4}
\end{figure} 
\noindent
On using Eq.~(\ref{lemma2}) to express the three commutators on the RHS of this equation individually and
with $\epsilon_{A}\geq \Delta N_A,~\eta_{B}\geq \Delta D_B$ and some properties of inequalities,
we get the inequality stated in Eq.~(\ref{modiozawapdt}).  
\end{proof}
In the case of the inequality given in Eq.~(\ref{modiozawapdt}) for comparison we denote
\begin{eqnarray}
L_{New}^{(2)}&=& \epsilon_{A}\eta_{B} + \eta_{B}\Delta A + \epsilon_{A}\Delta B \nonumber \\
&&-\frac{1}{2} \frac{|\langle\Psi | N_{A}\Delta D_{B}  \pm i D_{B} \Delta N_{A}|\Psi^{\perp} \rangle |^{2} }{\epsilon_{A}\eta_{B}} \nonumber \\
&&-\frac{1}{2} \frac{|\langle\Psi | A \Delta D_{B}  \pm i D_{B} \Delta A|\Psi^{\perp} \rangle |^{2} }{\Delta A\eta_{B}} \nonumber \\
&&-\frac{1}{2} \frac{|\langle\Psi | N_{A}\Delta B  \pm i B\Delta N_{A}|\Psi^{\perp} \rangle |^{2} }{\epsilon_{A}\Delta B}.
\label{L New 2} 
\end{eqnarray}
For illustration, consider a qubit state, with $|\psi\rangle_{s} =\cos\theta|0\rangle + \sin\theta|1\rangle$ 
and $|\phi\rangle_{p}=|1\rangle$. We choose the observables
$M_{in}$, $B_{in}$ and $U$ as in previous case but the observable $A_{in}$ 
is defined (scaled down) as
\begin{equation}
A_{in} = \lambda~ (\sigma_{x} \otimes \mathbb{I}),
\end{equation}
with $\lambda=0.01$. The resulting plot has been depicted in Fig. \ref{Fig_4}. It is seen that for a  small number of states,
the Branciard bound is superseded by our new bound. This is remarkable because Branciard's error-disturbance relation reduces to Ozawa's error-
disturbance relation under very strong conditions, {\it e.g.}, either $\epsilon_{A}$ or $\eta_{B}$ 
must be zero and expectation value of commutator of $A$ and $B$ should vanish.
Otherwise the Branciard bound is much tighter than the Ozawa bound. However, this new bound tightening the Ozawa bound goes beyond the Branciard bound for a small fraction 
of states, {\it i.e.}, $1.5 \%$ even though none of the above conditions are satisfied. One may think that the comparison with the Branciard bound in this way is problematic, 
since the commutator still appears in the definition (\ref{L Branciard}). 
If one rewrites the Branciard relation differently giving a different upper bound on the commutator, then that may lead to a different result in comparison to the new bound. 
This would be explored in the future.\\

Recently, quantum uncertainty equalities were introduced by Yao {\it et al.} \cite{Yao} for the sum of variances 
$\Delta A^2+\Delta B^2$ and product of variances $\Delta A^2 \Delta B^2$ on the trend of stronger 
uncertainty relations \cite{Macconepati}, for all pairs of incompatible observables $A$ and $B$. The sum of variance equality 
can be written for the noise and the disturbance as follows

\begin{equation}
\Delta N_{A}^{2} + \Delta D_{B}^{2}=\pm i \langle\Psi|[N_{A},D_{B}]|\Psi\rangle + \sum _{k=1}^{d-1} \vert\langle \Psi |N_{A} \pm iD_{B}|\Psi_k^{\perp} \rangle \vert^{2}, \nonumber \\
\end{equation} where $\{|\Psi\rangle, {|\Psi_k^{\perp}\rangle\}}_{k=1}^{d-1}$ form an orthonormal and complete basis in $d$ - dimensional Hilbert space.
This leads to another error-disturbance inequality, as given by
                              
\begin{eqnarray}
&&\epsilon_{A}^{2} + \eta_{B}^{2}\geq \pm i \langle \psi|[A,B]|\psi\rangle \mp i \langle\Psi|[M_{out},B_{in}]|\Psi\rangle \nonumber \\
&& \mp i \langle\Psi|[A_{in},B_{out}]|\Psi\rangle +\sum _{k=1}^{d-1} \vert\langle \Psi |N_{A} \pm iD_{B}|\Psi_k^{\perp} \rangle \vert^{2}.\nonumber \\
\label{newmdrequality}
\end{eqnarray}
Eq.~(\ref{newmdrequality}) is even tighter than Eq.~(\ref{firstmdr}). However, this error-disturbance inequality we will discuss in a separate paper.  
\end{section}
\begin{section}{Conclusion and Future Scope}
In this letter, we have proved a new error-disturbance relation. This shows no explicit dependence on the variances of original 
observables to be measured. We have demonstrated that this new error-disturbance relation given in Eq.~(\ref{firstmdreqn}) can give rise to better bound than 
the previously known inequalities for some initial state and measurement strategy. It is shown to give improvement over the Branciard bound
for a particular state that couples with the probe through a specific interaction. We have also proved a modified version of Ozawa's error-disturbance 
relation given by Eq.~(\ref{modiozawapdt}) and illustrated this for 
qubit states and some choices of scaled down values of operator $A_{in}$ to have tighter bounds. It exhibits tightening of the
Branciard bound for a very small number of states. However, it gives better bound than the Ozawa's bound in all cases.\\

Our method may be extended to the case of initial system state and/or probe state 
being mixed states and this leads to many possibilities about precision of measurement in the case of mixed state. It would also be 
interesting to see if the history of any prior interaction between the system and the probe has any effect on 
the error-disturbance relations. Uncertainty relations have applications in detection of 
entanglement. However, if we wish to experimentally perform realistic measurements on states 
in order to detect entanglement, it is important to know the corresponding error-disturbance inequalities rather than uncertainty relations. 
This formalism used here giving tighter bounds to these error-disturbance inequalities 
may be more efficient in detection of entanglement. These issues may be explored in future.
\end{section}

\begin{section}{Acknowledgment}
Authors thank  M. N. Bera and Uttam Singh for helpful discussions. CM and NS acknowledges research fellowship of 
Department of Atomic Energy, Govt of India. 
\end{section}

\bibliographystyle{h-physrev4}

\begin{thebibliography}{60}
\bibitem{Heisenbergoriginal} W. Heisenberg, \textit{Physical Principles of Quantum Theory} (Dover, New York, 1949).
\bibitem{Kennard} E. H. Kennard, "Zur Quantenmechanik einfacher Bewegungstypen", Zeitschrift f{\"u}r Physik \textbf{44}, 326 (1927).
\bibitem{Robertson} H. P. Robertson, Phys. Rev. \textbf{34}, 163 (1929).
\bibitem{Schroedinger} E. Schr{\"o}dinger, "Zum Heisenbergschen Unsch{\"a}rfeprinzip", Berliner Berichte, 296 (1930).
\bibitem{Busch} P. Busch, T. Heinonen, and P. J. Lahti, Phys. Rep. \textbf{452}, 155 (2007).
\bibitem{Hall1} M. J. Hall, Gen. Relativ. Gravit. \textbf{37}, 1505 (2005).
\bibitem{Hofman} H. Hofmann, T. Takeuchi, \pra~\textbf{68}, 032103 (2003).
\bibitem{Guehne} O. G{\"u}hne,  \prl~\textbf{92}, 117903 (2004).
\bibitem{Fuchs} C. A. Fuchs, A. Peres, \pra~\textbf{53}, 2038 (1996).
\bibitem{Macconepati} L. Maccone, A. K. Pati, \prl~\textbf{113}, 260401 (2014).
\bibitem{Arthur-kelly} E. Arthurs, J. L. I. Kelly,  Bell Syst. Tech. J. \textbf{44}, 725 (1965).
\bibitem{Arthur-Goodman} E. Arthurs, M. S. Goodman, \prl~\textbf{60}, 2447 (1988).
\bibitem{Ozawa} M. Ozawa, \pra~\textbf{67}, 042105 (2003).
\bibitem{Ozawa2} M. Ozawa, Ann. Phys. \textbf{311}, 350 (2004).
\bibitem{Rozema} L. A. Rozema \textit{et al.}, \prl~\textbf{109}, 100404 (2012).
\bibitem{Erhart} J. Erhart \textit{et al.}, Nature Physics~\textbf{8}, 185 (2012). 
\bibitem{Buschlahtiwerner} P. Busch, P. Lahti, and R. F. Werner, \prl~\textbf{111}, 160405 (2013).
\bibitem{Ozawaresolve} M. Ozawa, arXiv:1308.3540v1 (2013).
\bibitem{Buscemi} F. Buscemi, M.J.W. Hall, M. Ozawa, M.M. Wilde, \prl~\textbf{112}, {050401} (2014).
\bibitem{Busch01} P. Busch, P.J. Lahti, R.F. Werner, arXiv:1312.4393 (2013).
\bibitem{Busch02} P. Busch, P.J. Lahti, R.F. Werner, arXiv:1402.3102 (2014).
\bibitem{Busch1} P. Busch, P. Lahti, and R. F. Werner, Rev. Mod. Phys. \textbf{86}, 1261 (2014).
\bibitem{Busch2} P. Busch, P. Lahti, and R. F. Werner, J. Math. Phys.~\textbf {55}, 042111 (2014).
\bibitem{Busch3} P. Busch, P. Lahti, and R. F. Werner, \pra~\textbf{89}, 012129 (2014).
\bibitem{Busch4} P. Busch, N. Stevens, \prl~\textbf{114}, 070402 (2015).
\bibitem{Macccone} L. Maccone, Eur. Phys. Lett. \textbf{77}, 40002 (2007).
\bibitem{Hall} M. J. Hall, \pra~\textbf{69}, 052113 (2004).
\bibitem{Weston} M. M. Weston \textit{et al.}, \prl~\textbf{110}, 220402 (2013).
\bibitem{Branciard} C. Branciard, Proc. Natl. Acad. Sci.~\textbf{110}, 6742 (2013).
\bibitem{Qiao} J. Li, K. Du, and C. F. Qiao, \pra~\textbf{91}, 012110 (2015). 
\bibitem{Yao} Y. Yao, X. Xiao, X. Wang, and C. P. Sun,  arXiv:1503.00239v1 (2015).
\end{thebibliography}

\end{document}